\documentclass[11pt,onecolumn,draftclsnofoot]{IEEEtran}
\usepackage{xspace,amsmath,amssymb,amsfonts,epsfig,subfigure,syntonly}
\usepackage{cite,bm,color,url,textcomp,empheq,boxedminipage}
\usepackage{algorithmicx,algorithm}
\usepackage{epstopdf,makecell}
\usepackage{empheq}
\usepackage{pifont}
\usepackage{subeqnarray}
\usepackage{cases}
\usepackage{graphicx,graphics}  
\usepackage{multirow,multicol}
\usepackage{psfrag}    
\usepackage{stfloats}
\usepackage{url}
\usepackage[noend]{algpseudocode}

\newtheorem{lemma}{\underline{Lemma}}[section]

\newtheorem{remark}{\underline{Remark}}[section]

\newcommand{\mv}[1]{\mbox{\boldmath{$ #1 $}}}

\newtheorem{proof}{Proof}[section]

\psfull

\begin{document}
\title{Energy-Efficient Federated Edge Learning with Joint Communication and Computation Design}
\author{{Xiaopeng Mo and Jie Xu}
\thanks{ X. Mo is with the School of Information Engineering, Guangdong University of Technology, Guangzhou 510006, China. He is also with the Future Network of Intelligence Institute (FNii), The Chinese University of Hong Kong, Shenzhen, Shenzhen 518172, China (e-mail: xiaopengmo@mail2.gdut.edu.cn).}
\thanks{J. Xu is with the Future Network of Intelligence Institute (FNii) and the School of Science and Engineering, The Chinese University of Hong Kong, Shenzhen, Shenzhen 518172, China (e-mail:
xujie@cuhk.edu.cn). J. Xu is the corresponding author.}
}
\maketitle
\begin{abstract}
 This paper studies a federated edge learning system, in which an edge server coordinates a set of edge devices to train a shared machine learning (ML) model based on their locally distributed data samples. During the distributed training, we exploit the joint communication and computation design for improving the system energy efficiency, in which both the communication resource allocation for global ML-parameters aggregation and the computation resource allocation for locally updating ML-parameters are jointly optimized. In particular, we consider two transmission protocols for edge devices to upload ML-parameters to edge server, based on the non-orthogonal multiple access (NOMA) and time division multiple access (TDMA), respectively. Under both protocols, we minimize the total energy consumption at all edge devices over a particular finite training duration subject to a given training accuracy, by jointly optimizing the transmission power and rates at edge devices for uploading ML-parameters and their central processing unit (CPU) frequencies for local update. We propose efficient algorithms to optimally solve the formulated energy minimization problems by using the techniques from convex optimization. Numerical results show that as compared to other benchmark schemes, our proposed joint communication and computation design significantly improves the energy efficiency of the federated edge learning system, by properly balancing the energy tradeoff between communication and computation.
\end{abstract}
\begin{IEEEkeywords}
Federated edge learning, energy efficiency, joint communication and computation design, resource allocation, non-orthogonal multiple access (NOMA), optimization.
\end{IEEEkeywords}
\section{Introduction}
Artificial intelligence (AI) and machine learning (ML) have found abundant applications in, e.g., computer vision, recommendation systems, and natural language processing. The recent success of AI and ML depends on various factors such as the development of new algorithms (e.g., deep learning \cite{deep_learning}), the availability of massive data and the exponential increase of computation power. Normally, training proper AI/ML models requires huge computation power and massive training data. Therefore, the training of AI/ML models is conventionally implemented at the cloud or data center that has strong computation and storage capability \cite{cloud_resource}.

With recent technical advancements in Internet of Things (IoT) and fifth-generation (5G) celluar networks, massive data are generated by end devices (such as smart sensors, IoT devices, and smart phones) and  cellular base stations (BSs) at the wireless network edge. In this case, the conventional centralized cloud learning may not be applicable any longer, as the collection of big data from massive edge devices to the central cloud is costly and may cause extremely high traffic loads in communication networks. Motivated by the recent development of mobile edge computing (see, e.g., \cite{kaibing_survey}), a new paradigm of distributed mobile edge learning has been proposed (see, e.g., \cite{edge_learning,edge_learning2,edge_learning3,JCIN1}), in which massive edge devices are enabled to cooperate in training shared ML models by exploiting their locally distributed data and computation power. By pushing ML tasks from far-apart cloud to nearby edge, the mobile edge learning technique can significantly decrease the end-to-end latency and considerably reduce the traffic loads in communication networks. It is envisioned that the mobile edge learning will be a key technology in the beyond-fifth-generation (B5G) or sixth-generation (6G) cellular networks for enabling new applications such as autonomous driving, virtual reality (VR) and augmented reality (AR), thus realizing the network intelligence vision \cite{5g_beyond}.

Among different mobile edge learning approaches, federated edge learning is particularly appealing (see, e.g., \cite{fl,f2,f3,google_keyboard,FL_white_paper,data_gradient_compression,data_compression,S.Wang_Edge_learning,client_selection,low_latency_fl,over_the_air,over_the_air1,hierarchical_FL}), which allows a central node (such as an edge server) to coordinate a large number of edge devices to cooperate in training shared ML models based on their locally distributed data samples. Generally speaking, the objective of federated edge learning is to find optimized ML-parameters by minimizing the loss function via distributed optimization. In particular, the distributed gradient descent method has been widely adopted to solve the loss-function minimization problem, which is implemented in an iterative manner as follows. At each (outer) iteration, the edge server first broadcasts the global ML-parameters to edge devices, such that all participating edge devices can synchronize their  local ML-parameters; next, each edge device individually updates its local ML-parameters by computing the gradients based on their own data samples, where the local update is normally implemented over several (inner) iterations to speed up the convergence\cite{S.Wang_Edge_learning}; then, after local update, the edge devices upload their local ML-parameters to the edge server, such that the edge server can aggregate them to obtain an updated global ML-parameters. As the above procedures proceed, the edge server can obtain converged global ML-parameters that correspond to the desirable ML-model. For convenience, we refer to the above outer and inner iterations as global and local iterations, respectively. As no explicit data sharing from edge devices is required, the federated edge learning efficiently preserves the data privacy and security for edge devices. Notice that the ML-parameters are frequently exchanged between the edge server and edge devices, and as a result, the performance of federated edge learning is fundamentally constrained by the communication between the edge server and edge devices, especially when they are connected by wireless links that are unstable and may fluctuate significantly over time. This thus calls for a new design principle for federated edge learning in an interdisciplinary manner from both computer science and wireless communications perspectives.

In particular, the implementation of federated edge learning over wireless networks faces the following technical challenges. First, due to the involvement of both communication and computation, how to jointly  design them for optimizing the ML-training performance in terms of the training speed and accuracy is a new problem to be tackled. This problem is particularly difficult. This is due to the fact that the ML-training performance depends on many different factors (such as the employed ML algorithms and the data distribution among edge devices), and there does not exist a generic analytic relationship between the ML-training performance metric and the communication/computation parameters. Next, as the edge devices are normally powered by battery with finite sizes, their limited energy supply is another challenge to be dealt with, especially when the trained ML-model contains a large number of parameters, leading to heavy computation and communication loads. In addition, the federated edge learning faces the so-called straggler's dilemma issue, i.e., the ML-training performance is limited by the slowest edge devices in communication and computation\cite{edge_learning3}.

In the literature, there have been some prior works \cite{data_gradient_compression,data_compression,S.Wang_Edge_learning,client_selection,low_latency_fl,over_the_air,over_the_air1,hierarchical_FL} investigating the communication-constrained federated edge learning systems from different perspectives. For example, in \cite{data_gradient_compression}\cite{data_compression}, the authors proposed gradient compression methods to accelerate the training speed by reducing the required communication cost for exchanging gradients. \cite{S.Wang_Edge_learning} optimized the numbers of global and local iterations to maximize the ML-training accuracy, subject to the communication resource constraints. \cite{client_selection} considered the heterogeneity of wireless channels at different edge devices, based on which a client selection algorithm was proposed to improve the efficiency of ML-model training, in which only the edge devices with good communication and computation qualities are chosen for avoiding the straggler's dilemma. Furthermore, the so-called over-the-air computation approach\cite{kaibing_aircomp}\cite{xiaowen_aircomp} is utilized in the federated edge learning systems \cite{low_latency_fl,over_the_air,over_the_air1}, in which the superposition property of wireless multiple-access channels is exploited for speeding up the global ML-model aggregation. In addition, \cite{hierarchical_FL} further considered a hierarchical federated learning structure integrating devices, edge and cloud, in which multiple edge servers are allowed to perform partial ML-model aggregation to speed up the training process by reducing the direct communication rounds from the edge devices to the cloud. Despite such research progresses, however, the energy-efficient federated edge learning design by taking into account both communication and computation still remains a topic that is not well addressed. This thus motivates our investigation in this work.

 This paper considers a federated edge learning system consisting of one edge server and  multiple edge devices. With the coordination of the edge server, the edge devices use the distributed batch gradient descent (BGD) method to train a shared ML-model. Suppose that the ML-model training is subject to a given training delay requirement. Our objective is to minimize the energy consumption at edge devices during the training, by jointly designing their communication resource allocation (i.e., the transmission power and corresponding rates) for global ML-parameters uploading and aggregation, and computation resource allocation (i.e., the central frequency unit (CPU) frequencies) for local ML-parameters updates.

 In particular, we present both communication and computation energy consumption models, which are functions with respect to the numbers of local and global iterations, and the communication and computation loads at each iteration. We also consider two transmission protocols for the edge devices to upload their local ML-parameters to edge server, namely non-orthogonal multiple access (NOMA) and time division multiple access (TDMA), respectively. Under the two transmission protocols, the formulated energy minimization problems are both non-convex and difficult to solve in general. To tackle this issue, we transform them into convex forms and then present efficient algorithms to solve them optimally. Numerical results reveal interesting tradeoffs among energy consumption, training speed and training accuracy. It is shown that our proposed joint communication and computation design achieves significant performance gains over other benchmark schemes without such joint optimization. It is also shown that properly choosing the numbers of global and local iterations can efficiently balance the communication-computation tradeoff for further enhancing the system energy efficiency.

Note that to our best knowledge, there have been two prior works \cite{energy_efficient_radio}\cite{energy_efficient_walid_saad} considering the energy-efficient communication design for federated edge learning systems, which are different from this paper in the following aspects. First, \cite{energy_efficient_radio} only focused on the communication resource allocation, while this paper considers both communication and computation resource allocations. Next, different from \cite{energy_efficient_walid_saad} that aimed to minimize the weighted sum of training delay and energy consumption, this paper aims to minimize the energy consumption subject to training delay constraints. While \cite{energy_efficient_walid_saad} used a (loose) upper bound of the loss function to represent the training accuracy and to determine the numbers of local and global iterations, this paper uses simulations to accurately reveal the effects of the numbers of local and global iterations on training accuracy, and further shows the communication-computation energy tradeoff in choosing these parameters. Also note that the joint communication and computation resource allocations have been extensively investigated in the MEC literature (see, e.g., \cite{mec_energy_efficient,mec_energy_efficient1,wangfeng_mec,caoxiaowen_mec,noma_trans_time,JCIN2}). Nevertheless, the MEC studies normally focused on the task execution delay as the performance metric under general
computation task models; while in the federated edge learning in this paper, we are interested in specific computation tasks for training ML-models, for which the training speed and accuracy are used as key performance measures.

The remainder of this paper is organized as follows. Section II presents the system model of our considered federated edge learning system. Section III formulates the joint communication and computation resource allocation problems for minimizing the total energy consumption at edge devices. Sections IV and V propose optimal solutions to the energy minimization problems under the NOMA and TDMA transmission protocols, respectively. Section VI provides numerical results to validate the effectiveness
of our proposed joint designs. Finally, Section VII concludes this paper.
\begin{figure}
  \centering
 \vspace{0em}
  \includegraphics[width=11cm]{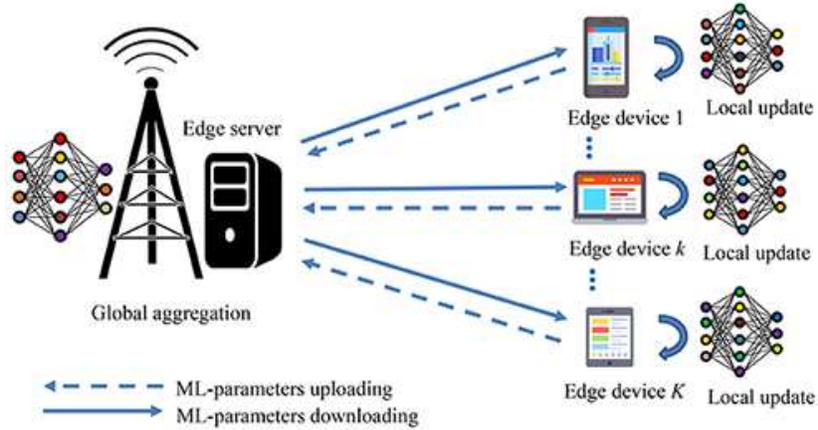}
    \small\caption{\small Illustration of the federated edge learning system with one edge server coordinating multiple edge devices to train shared ML-models.}
    \label{fig1}
     \vspace{0em}
\end{figure}
\section{System Model}
In this work, we consider a federated edge learning system consisting of an edge server and a set $\mathcal{K}\triangleq\{1,...,K\}$ of edge devices, as shown in Fig. \ref{fig1}. The edge server and edge devices are each deployed with one single antenna. In this system, the edge server coordinates $K$ edge devices to train a shared ML-model (such as linear regression, support vector machine (SVM), and deep neural networks
(DNN)) by using their local data. Let $\mv w$ denote the ML-parameters in the ML model to be trained, $\mathcal{D}_k$ denote the set of data samples distributed at edge device $k$. Accordingly, we use $f_i(\mv w)$ to denote the loss function
for each data sample $i\in\mathcal{D}_k$. For instance, for linear regression, by letting $\mv x_i$ and $y_i$ denote the input vector and the desired output scalar for each data sample $i$, then the loss function can be expressed as $f_i(\mv w) = \frac{1}{2}| y_i-\mv w^T \mv x_i|^2$, where $|\cdot|$ denotes the absolute value of a real number.\footnote{ Please refer to \cite{loss_function1,loss_function2} for details about the loss functions for SVM, K-means and convolutional neural network (CNN).} In this case, the average loss function at edge device $k\in \mathcal{K}$ is given by
\begin{align}
F_k(\mv w)\triangleq\frac{1}{|\mathcal{D}_k|}\sum_{i\in\mathcal{D}_k}f_i(\mv w),\label{local_loss}
\end{align}
where $|\mathcal{D}_k|$ denotes the cardinality of set $\mathcal{D}_k$. Accordingly, the average loss function for all the $K$ edge devices is given as
\begin{align}
F(\mv w)=\frac{\sum_{k=1}^{K}|\mathcal{D}_k|F_k(\mv w)}{\sum_{k=1}^{K}|\mathcal{D}_k|}.
\end{align}
In federated edge learning, the objective is to find desirable ML-parameters that minimize the average loss function $F(\mv w)$, i.e.,
\begin{align}
\mv w^*= \arg\min\limits_{\mv w} F(\mv w).\label{loss_function}
\end{align}
 Towards this end, we use the distributed BGD that is implemented in a distributed manner. Suppose that $M$ and $N$ denote the numbers of global and local iterations during the distributed BGD, respectively. For each global iteration $i\in\{1,...,M\}$, let $\mv w^{(0)}$ denote the initial global ML-parameters at the edge server, $\mv w^{(i)}$ denote the global ML-parameters at the edge server after the $i$-th global iteration, and $\mv w^{(i,j)}_k$ denote the local ML-parameters at edge device $k$ after the $j$-th local iteration of the $i$-th global iteration. We also use $S$ to denote the required bits for sending $\mv w_k^{(i,N)}$ to the edge server, which generally depends on the quantization and compression methods used.
Then, the following procedure is implemented at each global iteration $i$.
\begin{enumerate}
  \item {\it Edge server broadcasts global ML-parameters to devices}: The edge server broadcasts the global model ML-parameters $\mv w^{(i-1)}$ to the $K$ edge devices, where the local ML-parameters at each edge device $k$ are set as $\mv w^{(i-1)}$, i.e., $\mv w_k^{(i,0)} = \mv w^{(i-1)}, \forall k\in \mathcal K$.
  \item  {\it  Each edge device locally updates ML-parameters iteratively}: Each edge device $k\in \mathcal{K}$ updates its local ML-parameters in an iterative manner based on the gradient of average local loss function $F_k(\mv w)$. In each local iteration $j\in\{1,...,N\}$, supposing that the gradient of $F_k(\mv w)$ at point $\mv w$ is expressed as $\nabla F_k(\mv w)$, then we have $\mv w_k^{(i,j)}= \mv w_k^{(i,j-1)}+\eta \nabla F_k(\mv w_k^{(i,j-1)})$, where $\eta$ denotes the learning step size.
  \item  {\it Edge devices upload local ML-parameters to server}:  The $K$ edge devices upload their updated local ML-parameters $\mv w_1^{(i,N)},...,\mv w_K^{(i,N)}$ to the edge server.
  \item  {\it Edge server updates global ML-parameters via aggregation}: The edge server aggregates all the uploaded local ML-parameters $\mv w_k^{(i,N)}$ from $K$ edge devices and updates the global ML-parameters by averaging them, i.e., $\mv w^{(i)}=\frac{\sum_{k\in\mathcal{K}}\mv w^{(i, N)}_k}{K}$.
\end{enumerate}
  After the $M$ global iterations, the global ML-parameters $\mv w^{(M)}$ at the edge server are set as the desirable solution to problem (\ref{loss_function}), i.e., $\mv w^* \gets \mv w^{(M)}$. Notice that the performance of the federated edge learning depends on the numbers of global and local iterations $M$ and $N$. In general, larger values of $M$ and $N$ lead to smaller loss function and higher training accuracy, but also result in larger communication and computation energy consumption as well as longer training delay. Furthermore, to achieve the same training accuracy, it is possible to increase the number of local iterations $N$ with higher {\it computation} energy consumption and delay for trading for a smaller number of global iterations $M$ with lower {\it communication} energy consumption and delay. Therefore, there generally exist interesting performance tradeoffs in choosing $M$ and $N$ to balance among the training speed, accuracy, and energy consumption. We will leave the detailed discussion on choosing $M$ and $N$ in section VI, and will focus on the joint communication and computation resource allocation under given $M$ and $N$ in Sections III-V unless stated otherwise.

  Specifically, in this paper, we focus on the energy-efficient operation of energy-hungry edge devices for federated edge learning, and as a result, we omit the communication and computation energy consumptions at the edge server, as well as the time delay for global ML-parameters broadcasting and aggregation in steps 1) and 4), respectively.

In the following, we explain the computation energy consumption at edge devices for updating local ML-parameters in step 2) and the communication energy consumption for uploading local ML-parameters to the edge server in step 3), respectively.

\subsection{Local ML-Parameters Update at Edge Devices}
First, we consider the local ML-parameters update at one particular global iteration $i$, in which each edge device $k$ updates its local ML-parameters $N$ times, by computing  $\mv w_k^{(i,j)}= \mv w_k^{(i,j-1)}+\eta \nabla F_k(\mv w_k^{(i,j-1)}),\forall j\in\{1,...,N\}$. In order to model the energy consumption and delay for such computation, we use the number of floating point operations (FLOPs) to measure the computation complexity. For analytical convenience, it is assumed that the number of FLOPs\footnote{ How to obtain an accurate estimate of the number of FLOPs for each local update is an important task. Here, we explain how to estimate that for CNN, while those for other ML models (such as SVM and DNN) can be similarly obtained. In particular, a CNN  generally consists of multiple fully-connected layers and convolutional layers. For each fully-connected layer, the number of FLOPs can be easily evaluated as the product of the inputs and outputs. For each convolutional layer, the number of FLOPs can be estimated as the product of various factors including the spatial width and height of the feature map, previous and current layers, and the width and height of the Kernel. Please refer to \cite{FLOPs} for more details and examples.} needed for each data sample per local update is a constant, denoted by $a$. As a result, the total number of FLOPs required at edge device $k\in\mathcal{K}$ is given as $F_k = a\times |\mathcal{D}_k|$. In order to efficiently reduce the energy consumption at each edge device, the dynamic voltage and frequency scaling (DVFS) technique \cite{DVFS} is applied to adaptively adjust the CPU frequency to perfectly match the computation demand. For each edge device $k\in\mathcal{K}$, let $C_k$ denote the number of FLOPs within a CPU cycle and $f_{k}$ denote the CPU frequency of the whole operation duration, where $f_{k}\in(0,f_k^{\text{max}}]$, and $f_k^{\text{max}}$ denotes the maximum CPU frequency at edge device $k$. Accordingly, the computation time duration for each local update at edge device $k$ is given by
\begin{align}
t_k^{\text{loc}}=\frac{F_k}{C_k}\frac{1}{f_{k}}.
\end{align}

 For each edge device $k\in\mathcal{K}$, the energy consumption for local ML-parameters update mainly comes from that consumed by the CPU. It has been shown in \cite{kaibing_survey} that the CPU power consumption is proportional to the square of CPU chip's voltage $V^2$ and the operating CPU clock frequency $f$, where the voltage is approximately linear with respect to the CPU frequency\cite{power_linear_model}.
 Therefore, the total energy consumption at edge device $k\in\mathcal{K}$ for local ML-parameters update is given by\cite{mec_energy_efficient1}
\begin{align}
E_{k}^{\text{loc}}
=\frac{F_k}{C_k}\varsigma_kf_{k}^2,
\end{align}
where $\varsigma_k$ is a constant coefficient that depends on the chip architecture at edge device $k\in\mathcal{K}$.
\subsection{Local ML-Parameters Uploading from Edge Devices to Server}
 Next, we consider the local ML-parameters uploading from the $K$ edge devices to the edge server. For the purpose of initial investigation, we consider a quasi-static frequency non-selective channel model, in which the wireless channels are assumed to remain unchanged during the whole training process.
In particular, we consider two transmission protocols for the $K$ edge devices to upload their ML-parameters to the edge server, namely NOMA and TDMA, respectively.
\subsubsection{NOMA-Based Transmission}
First, we consider the NOMA-based transmission protocol, in which the $K$ edge devices simultaneously upload their local ML-parameters to the server. Let $t_{\text{up}}$ denote the transmission duration. Suppose that the transmit signal at edge device $k\in\mathcal{K}$ is $x_k$, which is a circularly symmetric complex Gaussian (CSCG) random variable with zero mean and unit variance. Then the received signal at the edge server is given by
\begin{align}
y=\sum_{k=1}^{K}\sqrt{p_k h_k}x_k+z,
\end{align}
where $p_k\geq0$ denotes the transmission power at edge device $k$, $h_k$ denotes the channel power gain from edge device $k$ to the edge server, and $z$ denotes the additive white Gaussian noise (AWGN) at the edge server that is a CSCG random variable with mean zero and variance $\sigma^2$, i.e., $z\sim \mathcal{C} \mathcal{N}(0,\sigma^2)$.

At the receiver side, the edge server adopts the minimum mean squared error successive interference cancellation (MMSE-SIC)\cite{noma_trans_time}\cite{Tse_Wireless}\cite{Tse_MAC} to decode information from the $K$ edge devices, where the edge server decodes the uploaded ML-parameters from the edge devices following a $\emph{successive decoding order}$, denoted by $\mv \pi$. Specifically, the edge server first decodes the information $x_{\pi(K)}$ from edge device $\pi(K)$, then decodes the information $x_{\pi(K-1)}$ by cancelling the interference from $x_{\pi(K)}$, followed by $x_{\pi(K-2)}$, ..., until decoding $x_{\pi(1)}$. Under Gaussian signalling and with a given decoding order $\mv \pi$ and transmission power $\{p_k\}$, the achievable rate (in bits per second) at edge device $\pi(k), k\in\mathcal{K}$, is given by\cite{Tse_Wireless,Tse_MAC}
\begin{align}
r_{\pi(k)}^{(\pi)}=B\log_2
\left(\frac{\sigma^2+\sum_{n=1}^{k}p_{\pi(n)} h_{\pi(n)}}
{\sigma^2+\sum_{n=1}^{k-1}p_{\pi(n)} h_{\pi(n)}}\right),\label{rate}
\end{align}
where $B$ denotes the transmission bandwidth. By properly designing the decoding order and employing time-sharing among different decoding orders, the achievable rate region for the $K$ edge devices (or equivalently the capacity region of the multiple-access channel) is given by\cite{Tse_Wireless,Tse_MAC}
\begin{align}
\mathcal{R}_{\text{NOMA}}(\mv p)=\left\{\mv r \in\mathbb{R}_+^{K\times1}\big|
\sum_{k\in\mathcal{\bar{K}}} r_k\leq B \log_2\big(1+\frac{\sum_{i\in\mathcal{\bar{K}}}p_k h_k}{\sigma^2}\big),\forall\mathcal{\bar{K}}\subseteq\mathcal{K}\right\},
\end{align}
where $\mv r \triangleq [r_1,...,r_K]^\dag$, $\mv p \triangleq [p_1,...,p_K]^\dag$, and $\mathbb{R}_+^{K\times1}$ denotes the set of all non-negative real vectors with dimension $K$.
 Here, the superscript $\dagger$ denotes the transpose. Supposing that each edge device $k\in\mathcal{K}$ needs to accomplish its ML-parameters uploading task within the duration $t_{\text{up}}$ that is to be optimized, we can obtain the following transmission constraint
\begin{align}
r_k t_{\text{up}}\geq S, \forall k\in\mathcal{K},
\end{align}
where recall that $S$ is the number of bits required for each edge device to send the local ML-parameters to edge server.
 The communication energy consumption at edge device $k$ for uploading is thus given by $E_k^\text{NOMA}=p_k t_{\text{up}}, k\in\mathcal{K}$.

 By combining the time delay for both local ML-parameters updates and uploading, and ignoring the delay for global ML-parameters broadcasting and aggregation, the total delay for training the ML-model is thus
\begin{align}
T^\text{NOMA}=M(N t_{\text{loc}}+t_{\text{up}}),
\end{align}
where the communication for uploading is implemented over $M$ rounds (or global iterations), and the computation for local updates at each device is implemented $M\times N$ times in total.
\subsubsection{TDMA-Based Transmission}
Next, we consider the TDMA-based transmission protocol, where the $K$ edge devices upload their updated local ML-parameters to the edge server over orthogonal time resources. Let $t_k^{\text{up}}$ denote the local ML-parameters uploading duration allocated for edge device $k \in \mathcal{K}$. By letting $p_k$ denote the transmission power at edge device $k$, the achievable rate of this device is given as
$B\log_2(1+\frac{p_kh_k}{\sigma^2}), k\in\mathcal{K}$.
In order for each edge device $k$ to successfully upload $S$ bits to edge server over duration $t_k^{\text{up}}$, we have
\begin{align}
B\log_2\big(1+\frac{p_kh_k}{\sigma^2}\big)t_k^{\text{up}}\geq S, \forall k\in\mathcal{K}.
\end{align}
 The communication energy consumption at each edge device $k$ for uploading is thus given by $E_k^\text{TDMA}=p_k t_k^{\text{up}}$. Accordingly, the total training delay is given by
 \begin{align}
 T^\text{TDMA}=M(N t_{\text{loc}}+\sum_{k=1}^K t_{k}^{\text{up}}).
 \end{align}
\section{Problem Formulation}
 Our objective is to minimize the total energy consumption at the $K$ edge devices, subject to a maximum training delay constraint $T$. The decision variables include the communication and computation resource allocation, i.e., the transmission power $\mv p$ and rate $\mv r$, and the CPU frequency $\mv f\triangleq[f_{1},...,f_{K}]^\dag$. In Sections III-V, in order to focus our study on the communication and computation resource allocations, we fix the numbers of global and local iterations $M$ and $N$. Note that the values of $M$ and $N$ are properly chosen in order to balance the communication-computation energy tradeoff while ensuring a certain training accuracy, as will be shown in Section VI.

 First, we consider the case with NOMA transmission protocol. In this case, the training-delay-constrained energy minimization problem is formulated as
\begin{align}
\text{(P1)}:\min_{\mv r,\mv f,\mv p, t_{\text{loc}}, t_{\text{up}}} &M\sum_{k=1}^{K}\left(N\frac{F_k}{C_k}\varsigma_kf_{k}^2
+p_k t_{\text{up}}
\right)
\nonumber\\
\text{s.t.}~~
&\mv r \in \mathcal{R}_{\text{NOMA}}(\mv p)\label{NOMA}\\
&r_k t_{\text{up}}\geq S, \forall k\in\mathcal{K}\label{transmitrate}\\
&M(N t_{\text{loc}}+t_{\text{up}})\leq T\label{timecontraint} \\
&t_{\text{loc}}\geq\frac{F_k}{C_k}\frac{1}{f_{k}},\forall k\in\mathcal{K}\label{max_t_loc}\\
&0\leq f_{k}\leq f_k^{\text{max}},\forall k\in\mathcal{K}\label{frequency}\\
&0\leq p_k \leq P_{\text{max}},\forall k\in\mathcal{K} \label{powercontraint},
\end{align}
where constraints (\ref{NOMA}) and (\ref{transmitrate}) ensure the uploading of local ML-parameters at the $K$ edge devices, constraint (\ref{timecontraint}) specifies the training delay requirement, and constraint (\ref{max_t_loc}) characterizes the computation requirements for local updates at the $K$ devices.

Next, we consider the case with TDMA transmission protocol. In this case, the training-delay-constrained energy minimization problem is formulated as
\begin{align}
\text{(P2)}:\min_{\mv f,\mv p, t_{\text{loc}},\{t_{k}^{\text{up}}\}} &M\sum_{k=1}^{K}\left(N\frac{F_k}{C_k}\varsigma_kf_{k}^2
+p_k t_k^{\text{up}}
\right)
\nonumber\\
\text{s.t.}~~~
&B\log_2(1+\frac{p_kh_k}{\sigma^2}) t_{k}^{\text{up}}\geq S, \forall k\in\mathcal{K}\label{TDMA}\\
&M(N t_{\text{loc}}+\sum_{k=1}^K t_{k}^{\text{up}})\leq T \label{tdma_delay}\\
&(\ref{max_t_loc}),(\ref{frequency}),(\ref{powercontraint}).\nonumber
\end{align}

Note that in both problems (P1) and (P2), there generally exists an tradeoff between  communication and computation in reducing the energy consumption. It is observed that the edge devices can increase the transmission power $\mv p$ (or {\it increase the communication energy consumption}) to reduce the communication time for uploading ML-parameters. Under the training delay constraint in (\ref{timecontraint}), this in turn allows these devices to increase the time duration for local updates with reduced CPU frequencies $\mv f$, thus leading to {\it reduced computation energy consumption}. Therefore, how to properly design $\mv p$, $\mv r$ and $\mv f$ to optimally balancing such a tradeoff is essential for minimizing the total energy consumption.

Also note that problems (P1) and (P2) are generally challenging to be solved as both of them are not convex due to the coupling of $p_k$ and $t_{\text{up}}$ / $t^{\text{up}}_k$. Furthermore, in problem (P1) the constraint in (\ref{NOMA}) corresponds to a number of ($2^{K}-1$) inequality constraints, thus making (P1) more difficult to be solved, especially when $K$ becomes large. Before proceeding to address problems (P1) and (P2) in Sections IV and V, respectively, we first check their feasibility in the following, i.e., checking whether these edge devices can efficiently accomplish the ML-model training task within the delay $T$.

\subsection{Feasibility Checking for Problem (P1)}
  Checking the feasibility of problem (P1) corresponds to showing whether these devices are able to accomplish the ML-model training within delay $T$. Therefore, this is equivalent to minimizing the total training delay by solving the following problem:
\begin{align}
\min_{\mv r, \mv f, \mv p, t_{\text{loc}}, t_{\text{up}}} &~~M(N t_{\text{loc}}+t_{\text{up}})\label{feasibility_noma}\\
\text{s.t.}~~
~~&(\ref{NOMA}),(\ref{transmitrate}),(\ref{max_t_loc}),(\ref{frequency}),(\ref{powercontraint}).\nonumber
\end{align}
It is evident that the minimum training delay is attained when the edge devices use the largest CPU frequency and the highest transmission power, i.e., $f_k=f_k^{\text{max}}, p_k=P_{\text{max}}, \forall k \in \mathcal{K}$. In this case, we have $t_{\text{loc}}=\max\limits_{k\in\mathcal{K}}\{\frac{F_k}{C_k}\frac{1}{f_k^{\text{max}}}\}$. However, it still remains to find the decoding orders at the edge server to determine the time delay for edge devices to upload the local ML-parameters to the edge server. This corresponds to solving the following problem:
\begin{align}
\min_{\mv r, \mv p, t_{\text{up}}} ~~~~&t_{\text{up}}\label{aa}\\
\text{s.t.}~~~~~&p_k=P_{\text{max}},\forall k \in \mathcal{K}\label{max_p}\\
&(\ref{NOMA}),(\ref{transmitrate}).\nonumber
\end{align}
It can be easily shown that solving problem (\ref{aa}) is equivalent to solving the following problem to maximize the minimum or common achievable communication rate among the $K$ edge devices, which has been optimally solved in \cite{noma_trans_time}.
\begin{align}
\max_{\mv r, \mv p, \bar r} &~~~\bar r\label{common rate}\\
\text{s.t.}~~
&r_k\ge \bar r\nonumber\\
&(\ref{NOMA})(\ref{max_p}).\nonumber
\end{align}
Let the optimal solution to problem (\ref{common rate}) as $\bar r^*, r_k^*=\bar r^*$, and $p_k^*=P_{\text{max}}, \forall k\in\mathcal{K}$. Then the  minimum communication delay is given as $t_{\text{up}}=\frac{S}{\bar r^*}$. Accordingly, we obtain the minimum training delay to problem (\ref{feasibility_noma}) as
 \begin{align}
 T_{\text{min}}^{\text{NOMA}}=M\left(N\times\max_{k\in\mathcal{K}}\big\{\frac{F_k}{C_k}\frac{1}{f_k^{\text{max}}}\big\}+\frac{S}{\bar r^*}\right).
\end{align}
If $ T_{\text{min}}^{\text{NOMA}}\leq T$, then edge devices are able to accomplish the communication task within duration $T$, i.e., problem (P1) is feasible. Otherwise, problem (P1) is infeasible.
\subsection{Feasibility Checking for Problem (P2)}
Similar as in Section III-A, in order to check the feasibility of problem (P2), we minimize the
training duration problem under the TDMA transmission protocol, for which the problem is formulated as
\begin{align}
\min_{\mv r,\mv f,\mv p, t_{\text{loc}}, \{t_{k}^{\text{up}}\}} &M(Nt_{\text{loc}}+\sum_{k=1}^Kt_{k}^{\text{up}})\label{c}\\
\text{s.t.}~~
&(\ref{max_t_loc}),(\ref{frequency}),(\ref{powercontraint}),(\ref{TDMA}).\nonumber
\end{align}
It is easy to show that the minimum training duration is attained when all $K$ edge devices use the maximum CPU frequency and the maximum transmission
power, i.e., $f_k=f_k^{\text{max}}, p_k=P_\text{max}, \forall k \in \mathcal{K}$. Hence, the minimum training duration under the TDMA transmission protocol is given by
\begin{align}
T_{\text{min}}^{\text{TDMA}}=M\left(N\times\max_{k\in\mathcal{K}}\big\{\frac{F_k}{C_k}\frac{1}{f_k^{\text{max}}}\big\}+\sum_{k\in\mathcal{K}}\frac{S}{B\log_2(1+\frac{P_{\text{max}}h_k}{\sigma^2})}\right).
\end{align}
It is thus concluded that if $T_{\text{min}}^{\text{TDMA}}\leq T$, then problem (P2) is feasible. Otherwise, problem (P2) is infeasible.

\begin{remark}\label{outperform} Notice that the achievable rate region under the NOMA transmission is always superior to that under the TDMA transmission\cite{Tse_MAC}. Therefore, the NOMA transmission protocol always leads to lower minimum training delay than the TDMA counterpart. Furthermore, it can be shown that any feasible solution of transmission power/rate and CPU frequencies to problem (P2) under the TDMA case is also an feasible solution to (P1) under the NOMA case, but the opposite may not be true. Therefore, it is expected that the NOMA transmission also achieves lower energy consumption than the TDMA transmission, as will be validated in Section VI.
\end{remark}
\section{Optimal Solution to Problem (P1) under NOMA}
 In this section, we propose an efficient algorithm to solve problem (P1) optimally. We first transform problem (P1) into a convex form and then obtain the optimal solution by employing the Lagrange duality method.
\subsection{Transformation of Problem (P1) into Convex Form}
 We first deal with the non-convexity of problem (P1).
 Towards this end, we introduce $e_k=p_kt_{\text{up}},$ and $s_k= r_k t_{\text{up}}, \forall k\in\mathcal{K}$, and accordingly define $\mv e\triangleq[e_1,...,e_K]^\dag$ and $\mv s\triangleq[s_1,...,s_K]^\dag$. By replacing $p_k=\frac{e_k}{t_{\text{up}}}$ and $r_k=\frac{s_k}{t_{\text{up}}}, \forall k\in \mathcal{K}$, we transform  problem (P1) into the following equivalent form:
 \begin{align}
\text{(P1.1)}:
\min_{ \mv s, \mv f,\mv e, t_{\text{loc}}, t_{\text{up}}} &M\sum_{k=1}^{K}\left(N\frac{F_k}{C_k}\varsigma_kf_{k}^2
+e_k\right)
\nonumber\\
\text{s.t.}~
&\mv s \in \mathcal{C}(\mv e, t_{\text{up}})  \label{transformNOMA}\\
&s_k\geq S,\forall k\in\mathcal{K}\label{a}\\
&0\leq e_k \leq P_{\text{max}}t_{\text{up}},\forall k\in\mathcal{K}\\
&(\ref{timecontraint}),(\ref{max_t_loc}),(\ref{frequency}),\nonumber
\end{align}
where
\begin{align}
\mathcal{C}(\mv e, t_{\text{up}})&=\left\{\mv s \in\mathbb{R}_+^{K\times1}\big|
\sum_{k\in\mathcal{\bar{K}}} s_k\leq B \log_2(1+\frac{1}{t_{\text{up}}}\frac{\sum_{k\in\mathcal{\bar{K}}}e_k h_k}{\sigma^2})t_{\text{up}},
\forall\mathcal{\bar{K}}\subseteq\mathcal{K}\right\}.\label{region}
\end{align}
 Notice that in (\ref{region}), the right-hand-side (RHS) of each inequality inside the set corresponds to a concave perspective function, and therefore, $\mathcal{C}(\mv e, t_{\text{up}})$ is a convex set. Accordingly, problem (P1.1) is a convex optimization problem. However, it is still intractable to solve problem (P1.1) via standard convex optimization techniques such as the interior point method\cite{convex_optimization}. This is due to the fact that constraint (\ref{transformNOMA}) represents a number of $(2^K-1)$ inequality constraints, thus making (P1.1) extremely difficult to be solved when $K$ is sufficiently large.
\subsection{Optimal Solution to Problem (P1.1) or (P1)}
As problem (P1.1) is convex and satisfies the Slater's condition, the strong duality holds between this problem and its dual problem\cite{convex_optimization}. Therefore, we leverage the Lagrange duality method to obtain the optimal solution to problem (P1.1). Let $\lambda_k\geq0, \mu_k\geq0,k\in\mathcal{K}$, denote the dual variables associated with the $k$-th constraints in (\ref{a}) and (\ref{max_t_loc}), respectively. We define $\mv \lambda\triangleq[\lambda_1,...,\lambda_K]^\dag$ and $\mv \mu\triangleq[\mu_1,...,\mu_K]^\dag$. Let $\nu\geq0$ denote the dual variable associated with constraint (\ref{timecontraint}). The partial Lagrangian of problem (P1.1) is given by
\begin{align}
&\mathcal{L}_1(\mv f, \mv s, \mv e, t_{\text{loc}}, t_{\text{up}}, \mv \lambda, \mv \mu, \nu)\nonumber\\
&=\sum_{k=1}^{K}\left(MN\frac{F_k}{C_k}\varsigma_kf_{k}^2
+\frac{F_k}{C_k}\frac{1}{f_{k}}\mu_k\right)
+\sum_{k=1}^{K}\lambda_k(S-s_k)+t_{\text{loc}}(\nu MN-\sum_{k=1}^{K} \mu_k)\nonumber\\
&+M\sum_{k=1}^{K}e_k+\nu Mt_{\text{up}}-\nu T.\label{Lg}
\end{align}

Then the dual function of problem (P1.1) is
\begin{align}
g_1(\mv \lambda, \mv \mu, \nu)
=&\min_{\mv f, \mv s,\mv e, t_{\text{loc}}, t_{\text{up}}}\mathcal{L}_1(\mv f, \mv s, \mv e, t_{\text{loc}}, t_{\text{up}}, \mv \lambda, \mv \mu, \nu)\label{dual_function}\nonumber\\
\text{s.t.}~~~
&\mv s \in \mathcal{C}(\mv e, t_{\text{up}}) \nonumber\\
&0\leq f_{k}\leq f_k^{\text{max}},\forall k\in\mathcal{K}\nonumber\\
&0\leq e_k \leq P_{\text{max}}t_{\text{up}},\forall k\in\mathcal{K}.
\end{align}
\begin{lemma}\label{t_loc}
In order for the dual function $g_1(\mv \lambda, \mv \mu, \nu)$ to be bounded from below (i.e. $g_1(\mv \lambda, \mv \mu, \nu))>-\infty$), it must hold that $(\nu MN-\sum_{k=1}^{K} \mu_k)\geq 0$.
\end{lemma}
\begin{proof}
Suppose that  $(\nu MN-\sum_{k=1}^{K} \mu_k)<0$. Then by setting $t_{\text{loc}}\rightarrow\infty$, we have $g_1(\mv \lambda, \mv \mu, \nu)\rightarrow-\infty$. Therefore, this lemma is proved.
\end{proof}
Accordingly, the dual problem of (P1.1) is
\begin{align}
\text{(D1.1)}:\max_{\mv \lambda, \mv \mu, \nu}~~&g_1(\mv \lambda, \mv \mu, \nu)\nonumber\\
\text{s.t.}~~~
&\lambda_k\geq0,\forall k\in\mathcal{K}\\
&\mu_k\geq0,\forall k\in\mathcal{K}\\
&\nu\geq0\\
&\nu MN-\sum_{k=1}^{K} \mu_k\geq 0.
\end{align}
In the following, we first solve problem (\ref{dual_function}) under given feasible $\mv \lambda, \mv \mu, \nu$ to obtain $g_1(\mv \lambda, \mv \mu, \nu)$, and then find the optimal $\mv \lambda, \mv \mu, \nu$ to maximize $g_1(\mv \lambda, \mv \mu, \nu)$ by solving problem (D1.1).
\subsubsection{Obtaining $g_1(\mv \lambda, \mv \mu, \nu)$ by Solving Problem (\ref{dual_function}) Under Given $\mv \lambda, \mv \mu$ and $ \nu$}
First, we decompose problem (\ref{dual_function}) into the following $(K+2)$ subproblems.
\begin{align}
&\min_{0\leq f_k\leq f_{k}^{\text{max}}}
~~MN\frac{F_k}{C_k}\varsigma_kf_{k}^2+\frac{F_k}{C_k}\frac{1}{f_{k}}\mu_k,~~\forall k\in\mathcal{K}.\label{solve_frequency}\\
&~~~\min_{t_{\text{loc}}}~~~~~t_{\text{loc}}(\nu MN-\sum_{k=1}^{K} \mu_k).\label{max_f_dual}\\
&~~\min_{\mv s,\mv e, t_{\text{up}}}~~\sum_{k=1}^{K}\lambda_k(S-s_k)
+M\sum_{k=1}^{K}e_k
+\nu Mt_{\text{up}}\label{polymatroid}\\
&~~~~~\text{s.t.}~~~\mv s \in \mathcal{C}(\mv e, t_{\text{up}})\nonumber\\
&~~~~~~~~~~~~0\leq e_k \leq P_{\text{max}}t_{\text{up}},\forall k\in\mathcal{K}.\nonumber
\end{align}

First, we present the optimal solutions to the subproblems in (\ref{solve_frequency}), which are obtained  based on the Karush-Kuhn-Tucker (KKT) conditions \cite{convex_optimization}.
\begin{lemma}\label{f_opt}
The optimal solution $f_k^*,k\in\mathcal{K},$ to each subproblem in (\ref{solve_frequency}) is given by
\begin{align}
f_k^*=\min(\sqrt[3]{\frac{\mu_k}{2MN\varsigma_k}}\label{opt_f},f_k^{\text{max}}).
\end{align}
\end{lemma}
\begin{proof}
See Appendix A.
\end{proof}

 Next, as for subproblem (\ref{max_f_dual}), since $(\nu MN-\sum_{k=1}^{K} \mu_k)\geq 0$ must hold based on Lemma 4.1, we can obtain that the optimal $t_{\text{loc}}^*$ equals to zero when $(\nu MN-\sum_{k=1}^{K} \mu_k)>0$, and can be any arbitrary real number when $(\nu MN-\sum_{k=1}^{K} \mu_k)=0$.

Then, we solve problem (\ref{polymatroid}). Towards this end, we have the following lemma from \cite{Tse_MAC}, for which the proof is omitted for brevity.

\begin{lemma}
For any given $\mv \lambda$ and $\mv p$, the optimal solution of
\begin{align}
\max_{\mv r}~~~&\sum_{k=1}^{K}\lambda_k r_k\nonumber\\
\text{s.t.}~~~&\mv r \in \mathcal{R}_{\text{NOMA}}(\mv p) \label{lemma}
\end{align}
is attained at a point $\mv r^{(\pi)}\triangleq[r_{\pi(1)}^{(\pi)},...,r_{\pi(K)}^{(\pi)}]$ of the polymatroid $\mathcal{R}_{\text{NOMA}}(\mv p)$, where the successive decoding order $\mv\pi$ is any feasible permutation such that $\lambda_{\pi(1)}\geq...\geq\lambda_{\pi(K)}$. Meanwhile, the achievable rate of $r_{\pi(k)}^{(\pi)},k\in\mathcal{K}$ is given by (\ref{rate}).\label{lemma1}
\end{lemma}

Based on lemma \ref{lemma1}, it follows that under any given $\mv\lambda$, $\mv e$ and $t_{\text{up}}$, the optimal solution of
\begin{align}
\max_{\mv s}~~&\sum_{k=1}^{K}\lambda_k s_k\nonumber\\
\text{s.t.}~~~&\mv s \in \mathcal{C}(\mv e, t_{\text{up}})\label{s_b}
\end{align}
is attained at $\mv s^{(\pi)}\triangleq[s_{\pi(1)}^{(\pi)},...,s_{\pi(K)}^{(\pi)}]$, where $\mv\pi$ is any feasible permutation such that $\lambda_{\pi(1)}\geq...\geq\lambda_{\pi(K)}$, and $s_{\pi(k)}^{(\pi)}$ is given by
\begin{align}
 s_{\pi(k)}^{(\pi)}=B\log_2
\left(\frac{\sigma^2+\frac{1}{t_{\text{up}}}\sum_{n=1}^{k}e_{\pi(n)} h_{\pi(n)}}
{\sigma^2+\frac{1}{t_{\text{up}}}\sum_{n=1}^{k-1}e_{\pi(n)} h_{\pi(n)}}\right)t_{\text{up}}, k\in\mathcal{K}\label{solve_s}.
\end{align}

Based on (\ref{solve_s}), it follows that solving subproblem (\ref{polymatroid}) is equivalent to optimizing $\mv e$ and $t^{\text{up}}$ by solving the following problem:
\begin{align}
\max_{\mv e,t_{\text{up}}}~~
&\sum_{k=1}^{K} \left(\big(\lambda_{\pi(k)}-\lambda_{\pi(k+1)}\big)B \log_2\big(1+\frac{\sum_{n=1}^{k}e_{\pi(n)} h_{\pi(n)}}{t_{\text{up}}\sigma^2}\big)t_{\text{up}}
-M e_k\right)-\nu Mt_{\text{up}}\nonumber\\
\text{s.t.}&~~~~~~~~~~~~~~~~~~~~0\leq e_k \leq P_{\text{max}}t_{\text{up}},\forall k\in\mathcal{K},\label{e_t}
\end{align}
where we define $\lambda_{K+1}\triangleq 0$ for notational convenience.

Notice that problem (\ref{e_t}) is a convex optimization problem with respect to $\mv e$ and $t_{\text{up}}$, and thus can be solved efficiently by standard convex solvers, e.g., CVX\cite{cvx}.\footnote{Note that the optimal $\mv e^*$ (i.e., the optimal transmission power $\mv p^*$ in problem (P1)) is unique due to the strict convexity of problem (\ref{e_t}). However, the optimal $\mv s^*$ is generally non-unique\cite{wangfeng_mec}. Therefore, an addition step is needed later for constructing a feasible primal solution to problem (P1.1) by time-sharing among different solutions. Here, we can arbitrarily choose one solution of $\mv s^*$ for obtaining the dual function.} Let $\mv e^*$ and $t_{\text{up}}^*$ denote the optimal solution to problem (\ref{e_t}), and we obtain  $\mv s^*$ based on (\ref{solve_s}). Accordingly, $\mv s^*$, $\mv e^*$ and $t_{\text{up}}^*$ become the optimal solution to problem (\ref{polymatroid}).

\subsubsection{Finding Optimal $\mv\lambda$, $\mv\mu$ and $\nu$ to solve (D1.1)}
Next, we find the optimal ($\mv\lambda^{\text{opt}},\mv \mu^{\text{opt}},\nu^{\text{opt}}$) to maximize $g_1(\mv \lambda, \mv \mu, \nu)$ for solving the dual problem (D1.1). As the dual function $g_1(\mv \lambda, \mv \mu, \nu)$ is always convex but generally non-differentiable, we solve problem (D1.1)  by using the subgradient-based methods, such as the ellipsoid method\cite{ellipsoid}. Notice that the subgradient of the objective function $g_1(\mv \lambda, \mv \mu, \nu)$ with respect to $(\mv\lambda,\mv\mu,\nu)$ is $[s_1^*-S,...,s_K^*-S,
t^*_{\text{loc}}-\frac{F_1}{C_1}\frac{1}{f_1^*},...,t^*_{\text{loc}}-\frac{F_K}{C_K}\frac{1}{f_K^*},
T-M(Nt^*_{\text{loc}}+t_{\text{up}}^*)]$.
\subsubsection{Constructing Optimal Primal Solution to (P1.1) or (P1)}
Based on the optimal $\mv\lambda^{\text{opt}}$, $\mv\mu^{\text{opt}}$ and $\nu^{\text{opt}}$, we need to construct the optimal primal solution to (P1.1), denoted by ($\mv f^{\text{opt}}$, $\mv s^{\text{opt}}$, $\mv e^{\text{opt}}$, $t^{\text{opt}}_{\text{loc}}$, $t_{\text{up}}^{\text{opt}}$ ). By solving problem (\ref{e_t}) under $\mv \lambda^{\text{opt}}$ and $\nu^{\text{opt}}$, we obtain the optimal solution of $\mv e^{\text{opt}}$ and $t_{\text{up}}^{\text{opt}}$ to problem (P1.1). Similarly, by substituting $\mv\lambda^{\text{opt}}$ into (\ref{opt_f}), we  obtain the optimal $\mv f^{\text{opt}}$. Based on the obtained optimal $\mv f^{\text{opt}}$ and the constraint (\ref{max_t_loc}), we obtain that the optimal local update delay $t_{\text{loc}}^{\text{opt}}=\max\limits_{k\in\mathcal{K}}\{\frac{F_k}{C_k}\frac{1}{f_k^{\text{opt}}}\}$.

Finally, we still need to determine the primal optimal $\mv s^{\text{opt}}$ and the corresponding optimal decoding order at the edge server. Let $\mv \pi^{\text{opt}}=[\pi^{\text{opt}}(1),...,\pi^{\text{opt}}(K)]^\dag$ denote the permutation that satisfies the condition $\lambda_{\pi(1)}^{\text{opt}}\geq...\geq\lambda_{\pi(K)}^{\text{opt}}\geq0$. In this case, the primal optimal decoding order is $\mv \pi^{\text{opt}}$ and the optimal $\mv s^{\text{opt}}$ can be obtained based on (\ref{solve_s}). With optimal $\mv s^{\text{opt}}$ and $t_{\text{up}}$ at hand, the optimal transmission rate $\mv r^{\text{opt}}$ is obtained accordingly. It should be emphasized that if there exist some $\lambda_k^{\text{opt}}$'s that equal to each other, then we may need to {\it time-sharing} among different decoding orders. In this case we construct the feasible and optimal $\mv s^{\text{opt}}$ to the problem (P1.1) (or equivalently the optimal solution $\mv r^{\text{opt}}$ to problem (P1))  by using the {\it time-sharing} technique, similarly as that adopted in \cite{time_sharing}\cite{wangfeng_noma}. For brevity, we skip the discussion about the time-sharing technique here, and please refer to \cite{time_sharing}\cite{wangfeng_noma} for details.

\section{Optimal Solution to Problem (P2) under TDMA}
In this section, we obtain the optimal solution to
problem (P2) under the TDMA case. First, we transform problem (P2)
into a convex form.
Similarly as for solving problem (P1) in Section IV, we introduce a set of variables $\widetilde{\mv e}\triangleq[\widetilde{e}_1,...,\widetilde{e}_K]^\dag$ with $\widetilde{e}_k=p_kt_k^{\text{up}}, \forall k\in\mathcal{K}$. Accordingly, we transform problem (P2) into the following convex form:
\begin{align}
\text{(P2.1)}:
\min_{\mv f,\widetilde{\mv e}, t_{\text{loc}},\{t_{k}^{\text{up}}\}} &M\sum_{k=1}^{K}\left(N\frac{F_k}{C_k}\varsigma_kf_{k}^2
+\widetilde{e}_k
\right)
\nonumber\\
\text{s.t.}~~
&B\log_2(1+\frac{\widetilde{e}_kh_k}{t_k^{\text{up}}\sigma^2})t_k^{\text{up}}\geq S, \forall k\in\mathcal{K}\label{tdma_rate_constraint}\\
&0\leq \widetilde{e}_k \leq P_{\text{max}}t_k^{\text{up}},\forall k\in\mathcal{K}\label{tdma_powercontraint}\\
&(\ref{max_t_loc}),(\ref{frequency}),(\ref{tdma_delay}).\nonumber
\end{align}

Next, we employ
the Lagrange duality method to obtain the optimal solution to the convex problem (P2.1).
Let $\omega_k\geq 0, k\in\mathcal{K}$, and $\zeta\geq0$ denote
the dual variables associated with the constraints (\ref{max_t_loc}) and (\ref{tdma_delay}), respectively, where we define $\mv \omega\triangleq[\omega_1,...,\omega_K]^\dag$ for convenience. Then the partial Lagrangian of problem (P2.1) is given by
\begin{align}
\mathcal{L}_2(\mv \omega,\zeta, \mv f, \widetilde{\mv e},  t_{\text{loc}},\{t_{k}^{\text{up}}\})
&=\sum_{k=1}^K
(MN\frac{F_k}{C_k}\varsigma_kf_{k}^2+\omega_k\frac{F_k}{C_k}\frac{1}{f_k})
+(MN\zeta-\sum_{k=1}^K\omega_k)t_\text{loc}\nonumber\\
&+M\sum_{k=1}^{K}(\widetilde{e}_k+\zeta t_k^{\text{up}}).\label{tdma_lag}
\end{align}
The dual function of problem (P2.1) is defined as
\begin{align}
g_2(\mv \omega, \zeta)=&\min_{\mv f, \widetilde{\mv e},  t_{\text{loc}},\{t_{k}^{\text{up}}\}}
\mathcal{L}_2\big(\mv \omega, \zeta, \mv f, \widetilde{\mv e},  t_{\text{loc}},\{t_{k}^{\text{up}}\}\big)\label{tdma_dual_function}\\
\text{s.t.}~&~~~~(\ref{frequency}),(\ref{tdma_rate_constraint}),(\ref{tdma_powercontraint}).\nonumber
\end{align}
Similar as in Lemma \ref{t_loc}, it follows that in order for the dual function $g_2(\mv \omega, \zeta)$ to be bounded from below (i.e., $g_2(\mv \omega, \zeta)>-\infty$), we must have $MN\zeta-\sum_{k=1}^K\omega_k\geq 0$.
Accordingly, the dual problem of (P2.1) is given as
\begin{align}
\text{(D2.1)}:\max_{\mv \omega, \zeta}~~&g_2(\mv \omega, \zeta)\nonumber\\
\text{s.t.}~~~
&\omega_k\geq0,\forall k \in \mathcal{K}\\
&\zeta\geq0\\
&MN\zeta-\sum_{k=1}^K\omega_k\geq 0.
\end{align}

In the following, we first solve problem (\ref{tdma_dual_function}) under any given feasible $\mv \omega, \zeta$ to obtain $g_2(\mv \omega, \zeta)$, and then find the optimal $\mv \omega, \zeta$ to maximize $g_2(\mv \omega, \zeta)$ by solving problem (D2.1).

First, we  decompose problem (\ref{tdma_dual_function}) into $(2K+1)$ subproblems.
\begin{align}
\min_{ 0\leq f_k\leq f_k^{\text{max}}}&~~MN\frac{F_k}{C_k}\varsigma_kf_{k}^2+\omega_k\frac{F_k}{C_k}\frac{1}{f_k}, ~~k\in\mathcal{K}\label{tdma_fre}.
\end{align}
\begin{align}
\min_{t_\text{loc}}~~~t_\text{loc}(\zeta MN-\sum_{k=1}^{K}\omega_k)\label{tdma_t_loc}.
\end{align}
\begin{align}
\min_{\widetilde{e}_k, t_k^\text{up}}~~& \widetilde{e}_k+\zeta t_k^{\text{up}},~~k\in\mathcal{K}\label{tdma_t_up}\\
\text{s.t.}~~~&(\ref{tdma_rate_constraint}),(\ref{tdma_powercontraint}).\nonumber
\end{align}

Similar as in Lemma \ref{f_opt}, the optimal solution $f_k^\star$ to $k$-th subproblem in (\ref{tdma_fre}) is given by
\begin{align}
f_k^\star=\min\big(\sqrt[3]{\frac{\omega_k}{2MN\varsigma_k}},f_k^\text{max}\big).\label{tdma_f_opt}
\end{align}

For subproblem (\ref{tdma_t_loc}), since $(MN\zeta-\sum_{k=1}^K\omega_k)\geq 0$ must hold in subproblem (\ref{tdma_t_loc}), the optimal solution $t^\star_{\text{loc}}$  is zero when $(MN\zeta-\sum_{k=1}^K\omega_k)> 0$, and can be any real number when $(MN\zeta-\sum_{k=1}^K\omega_k)= 0$.

Then, consider subproblem (\ref{tdma_t_up}). It is noted that $B\log_2(1+\frac{\widetilde{e}_kh_k}{t_k^{\text{up}}\sigma^2})t_k^{\text{up}}=S, \forall k\in\mathcal{K}$ must hold at the optimal solution, and thus we have
\begin{align}
\widetilde{e}_k=\big(2^\frac{S}{Bt_k^{\text{up}}}-1\big)\frac{t_k^{\text{up}}\sigma^2}{h_k}, \forall k\in\mathcal{K}.\label{e_k}
\end{align}
By substituting (\ref{e_k}) into (\ref{tdma_t_up}) and after some simple manipulations, the $k$-th subproblem in (\ref{tdma_t_up}) reduces to the following problem:
\begin{align}
\min_{t_k^\text{up}}~~&\big(2^\frac{S}{Bt_k^{\text{up}}}-1\big)\frac{t_k^{\text{up}}\sigma^2}{h_k}+\zeta t_k^{\text{up}}\label{t_k}\\
\text{s.t.}~~~&t_k^\text{up}\geq\frac{S}{B\log(1+\frac{P_\text{max}h_k}{\sigma^2})},~k\in\mathcal{K}.\nonumber
\end{align}
Problem (\ref{t_k}) is a convex optimization problem to $t_k^\text{up}$. By setting the first-derivative of the objective function to be zero, we have
\begin{align}
2^\frac{S}{Bt_k^\text{up}}\frac{\sigma^2}{h_k}(1-\frac{S\ln2}{Bt_k^\text{up}})-\frac{\sigma^2}{h_k}+\zeta=0.\label{bisection}
\end{align}
Let $\tau_k^\star, k\in\mathcal K,$ denote the solution to the above equation in (\ref{bisection}), which can be easily obtained via a bisection search. Then the optimal solution $t_k^{\text{up}^\star}$ to problem (\ref{t_k}) is given as
\begin{align}
t_k^{\text{up}^\star}=\max\left(\tau_k^\star,\frac{S}{B\log(1+\frac{P_\text{max}h_k}{\sigma^2})}\right),\forall k\in\mathcal K.
\end{align}
Based on optimal $t_k^{\text{up}^\star}$'s and (\ref{e_k}), the optimal $\widetilde{e}_k^\star$'s are obtained accordingly. Therefore, the dual function $g_2(\mv \omega, \zeta)$ in (\ref{tdma_dual_function}) is finally obtained.

Next, with obtained $\mv f^\star$, $t_{\text{loc}}^\star$, $\{t_k^{\text{up}^\star}\}$ and $\widetilde{\mv e}^\star$ at hand, we  employ the ellipsoid method to find the optimal $\mv\omega^{\text{opt}}$ and $\zeta^\text{opt}$ to maximize $g_2(\mv \omega, \zeta)$ for solving problem (D2.1).

Finally, based on the obtained optimal $\mv \omega^{\text{opt}}$ and $\zeta^\text{opt}$, we need to construct the optimal primal solution to (P2.1), denoted by $\mv f^{\text{opt}},\widetilde{\mv e}^{\text{opt}}, \{t_k^{\text{up}^\text{opt}}\}$, and $t_\text{loc}^{\text{opt}}$. By substituting $\mv \omega^\text{opt}$ into (\ref{tdma_f_opt}), we can obtain the optimal $\mv f^\text{opt}$. Based on the obtained optimal $\mv f^\text{opt}$ and constraint (\ref{max_t_loc}), we obtain the optimal local update delay as $t_\text{loc}^\text{opt}
=\max\limits_{k\in\mathcal{K}}\{\frac{F_k}{C_k}\frac{1}{f_k^{\text{opt}}}\}$. By resolving equation (\ref{bisection}) under optimal $\zeta^\text{opt}$, we obtain the optimal $t_k^{\text{up}^\text{opt}}$, and thus $\widetilde{\mv e}^{\text{opt}}$ based on (\ref{e_k}) accordingly. Therefore, problem (P2.1) (or equivalently (P2)) is finally solved.

\section{Numerical Results}
In this section, we present numerical results to validate the performance of our proposed energy-efficient federated edge learning design. In the simulation, we consider the quasi-static channel model, where the path loss from the edge server to edge device $k\in\mathcal{K}$ is given by $\beta_0(\frac{d_k}{d_0})^{-\alpha_0}$. Here, $d_k$ denotes the corresponding distance, $\alpha_0=3$ denotes the path loss exponent, and $\beta_0=-30$ dB denotes the channel power gain at a reference distance of $d_0=1$ m. We set the system bandwidth for ML-parameters uploading as $B = 2$ MHz and the noise power at the edge server as $\sigma^2 = -100$ dBm. We consider the scenario with $K=3$ edge devices, where the distances from the edge devices to edge server are $d_1=100$ m, $d_2=150$ m and $d_3=200$ m, respective, unless stated otherwise. Furthermore, we consider that the ML-model is a CNN\footnote{We consider a similar CNN structure as in \cite{S.Wang_Edge_learning}, where $a\approx6$ GPLOPs based on \cite{FLOPs}. Furthermore, suppose that each parameter of the training model is quantized into $12$ bits, and as a result, we have $S\approx4.9$ Gbits accordingly.}, which is trained by using the distributed BGD method. We consider the MNIST data set, and all the three edge devices have $1000$ data samples in total.

 \begin{figure}
  \centering
  \includegraphics[width=10cm]{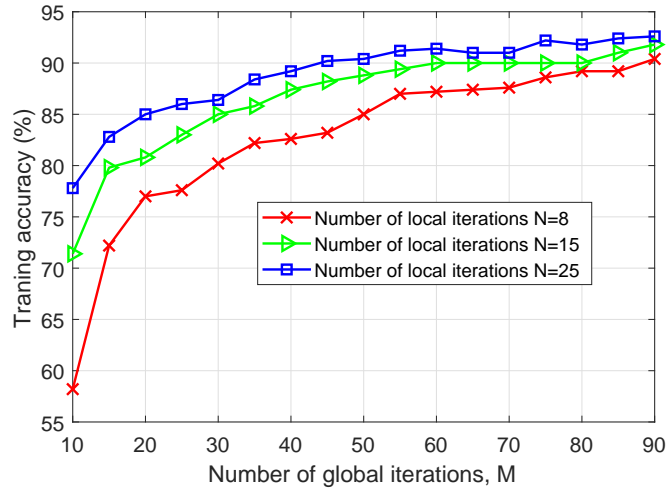}\\
  \small\caption{\small The average training accuracy versus the number of global iterations $M$ under different numbers of local iterations $N$.}\label{fig_accuracy}
\end{figure}
\begin{figure}
  \centering
  \includegraphics[width=10cm]{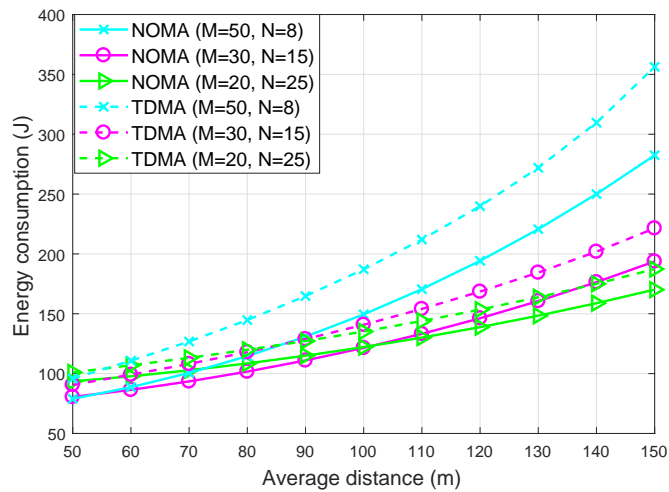}
    \small\caption{\small The energy consumption at edge devices versus the average distance between the edge server and edge devices, under different values of $M$ and $N$, where $T=166$ s.}
    \label{fig_distance}
     \vspace{0em}
\end{figure}
\begin{figure}
  \centering
  \includegraphics[width=10cm]{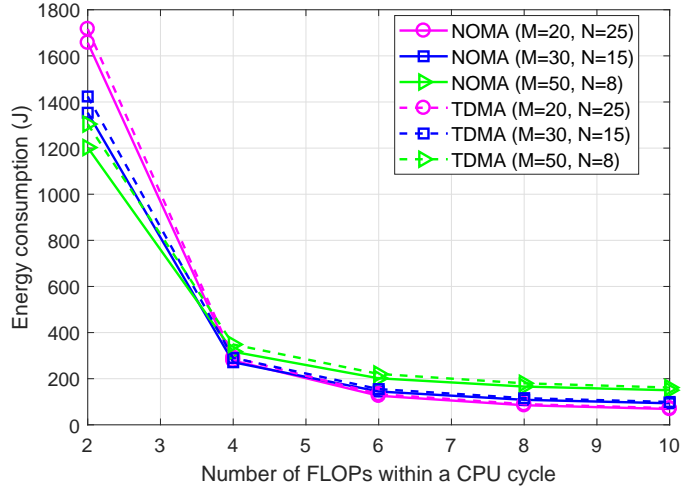}\\
  \small\caption{\small The energy consumption at edge devices versus the number of FLOPs within a CPU cycle $C=C_k, \forall k\in\mathcal{K}$, under different values of $M$ and $N$, where $T=850$ s.}\label{fig_c}
\end{figure}

Fig. \ref{fig_accuracy} shows the training accuracy versus the number of global iterations $M$, under different number of local iterations $N$. It is observed that under a given number of local (global) iterations $N$ $(M)$, the training accuracy increases as the number of global (local) iterations $M$ $(N)$ becomes larger. It is also observed that different values of $M$ and $N$ can be adopted to achieve the same training accuracy. For instance, to achieve the training accuracy of $85\%$, $(M,N)=(50,8)$, $(30,15)$ and $(25,20)$ can be adopted. This shows that more local iterations of computation (with larger $N$) can be used to trade for less global iterations of communication (with smaller $M$).

Next, we show the energy consumption of edge devices under different numbers of global and local iterations $M$ and $N$ to achieve the same training accuracy of $85\%$. Based on Fig. \ref{fig_accuracy}, we choose three pairs of parameters $(M, N)=(50,8)$, $(30,15)$ and $(20,25)$. Fig. \ref{fig_distance} shows the energy consumption at edge devices versus the average distance between the edge server and edge devices, under different values of $M$ and $N$. Here, the distances with  three edge devices are set as an arithmetic sequence, with the common difference being $45$ m (e.g., when the average distance is $100$m, then the distances from the three edge devices to the edge server are $d_1=55$ m, $d_2=100$ m and $d_3=145$ m, respectively). It is observed that the NOMA transmission always outperforms the TDMA. This is consistent with Remark \ref{outperform}.
It is also observed that under both NOMA and TDMA, when the average distance is short (e.g., shorter than $50$m), $(M=50, N=8)$ and $(M=30, N=15)$ lead to less energy consumption than $(M=20, N=25)$. This is due to the fact that in this case, the communication energy consumption for global ML-parameters uploading is small and thus the total energy consumption is dominated by the computation part for local ML-parameters update. Therefore, a small value of $N$ is desirable. By contrast, when the average distance increases (e.g., larger than $100$m), the opposite is observed to be true.

 Fig. \ref{fig_c} shows the energy consumption at edge devices versus the number of FLOPs within a CPU cycle $C=C_k, \forall k\in\mathcal{K}$, under different values of $M$ and $N$. It is observed that when the number of FLOPs within a CPU cycle is small or the computation capacities at edge devices are limited (e.g., $C_k=2$), the parameters of $(M=50, N=8)$ lead to lower energy consumption than that by $(M=30, N=15)$ and  $(M=20, N=25)$. This is due to the fact that under small value of $C$, higher CPU frequencies are generally needed for meeting the training deadline requirement, and as a result, the computation energy consumption for local ML-parameters update becomes the dominant part of the total energy consumption at edge devices. As a result, a smaller value of $N$ is preferred. By contrast, when $C$ becomes large (e.g., $C=10$), the parameters of $(M=20,N=25)$ are observed to outperform $(M=50,N=8)$ and $(M=30,N=15)$.

Then, we compare the performance of our proposed joint communication and computation design versus the following benchmark schemes. For comparison, in the following we fix the numbers of global and local iterations as $M=20$ and $N=25$, respectively, under the training accuracy requirement of $85\%$.
\begin{itemize}
  \item {\it Communication design only}: Each edge device $k\in\mathcal{K}$ locally updates their ML-parameters by using the maximum CPU frequency and only optimizes the transmission power and rate for global ML-parameters aggregation in the uploading process. Under NOMA and TDMA cases, the transmission powers and rates at the edge devices can be obtained by solving problems (P1) and (P2) under given $f_k=f_k^{\text{max}}, \forall k \in \mathcal{K}$, respectively.
  \item {\it Computation design only}: Each edge device $k\in\mathcal{K}$ updates the local ML-parameters to the edge server by using the maximum transmission power and only optimizes the CPU frequencies during the local ML-parameters update. Under NOMA and TDMA, the CPU frequencies at the edge devices can be obtained by solving problem (P1) and (P2) under given $p_k=P_{\text{max}}, \forall k \in \mathcal{K}$, respectively.
  \item {\it Training delay minimization}: The edge devices adopt the maximum transmission power and maximum CPU frequencies during the training process to minimize the delay for training the ML-model. This corresponds to solving problems (\ref{feasibility_noma}) and (\ref{c}) in Sections III-A and III-B, respectively.
\end{itemize}

Fig. \ref{max_fre_energy_compared} shows the energy consumption at edge devices versus the maximum CPU frequency $f^{\text{max}}=f_k^{\text{max}}, \forall k\in\mathcal{K}$, where $T=431$ s. It is observed that under both NOMA and TDMA, our proposed joint communication and computation designs outperform the other benchmark designs. This thus validates the benefit of our proposed designs. It is also observed that for our proposed designs under both NOMA and TDMA, the energy consumption at edge devices first decreases when $f^{\text{max}}$ increases from $0.5$ GHz to $1$ GHz, but keeps unchanged when it further increases. This is due to the fact that there exists an optimal CPU frequency between $0.5$ GHz to $1$ GHz, and therefore, further increasing the maximum CPU frequencies at edge devices cannot improve the energy efficiency. By contract, for the schemes with computation design only and training delay minimization, it is observed that increasing the maximum CPU frequency may lead to increased energy consumption. Furthermore, it is observed that for the two schemes with joint communication and computation design and communication design only, the NOMA transmission always outperforms the TDMA transmission in term of the energy efficiency; while for the other two schemes, NOMA may lead to higher energy consumption than TDMA.
\begin{figure}
  \centering
  \includegraphics[width=10cm]{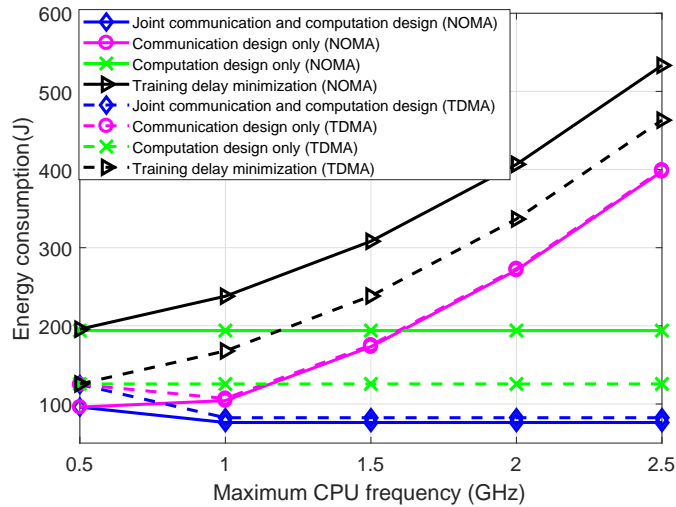}\\
  \small\caption{\small The energy consumption at edge devices versus the maximum CPU frequency $f^{\text{max}}=f_k^{\text{max}}, \forall k\in\mathcal{K}$, where $T=431$ s.}\label{max_fre_energy_compared}
\end{figure}

Fig. \ref{max_tr_energy_compared} shows the energy consumption at edge devices versus the maximum transmission power $P_\text{max}$, where $T=271$ s. Similar observations are made as in Fig. \ref{max_fre_energy_compared}.
\begin{figure}
  \centering
  \includegraphics[width=10cm]{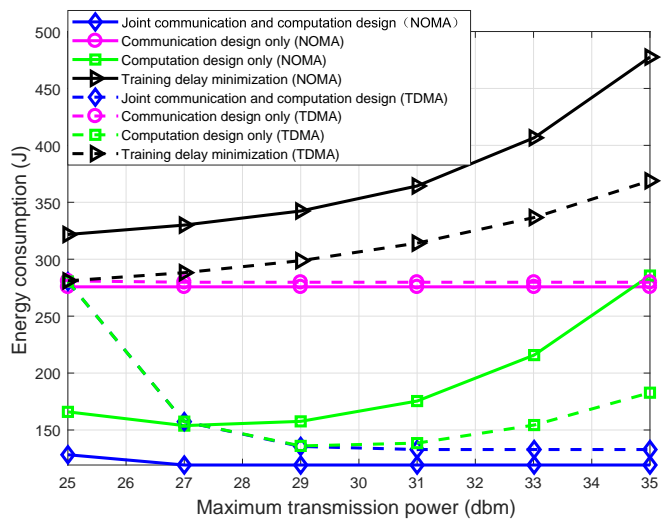}\\
  \small\caption{\small The energy consumption at edge devices versus the maximum transmission power $P_\text{max}$, where $T=271$ s.}\label{max_tr_energy_compared}
\end{figure}

\begin{figure}
  \centering
  \includegraphics[width=10cm]{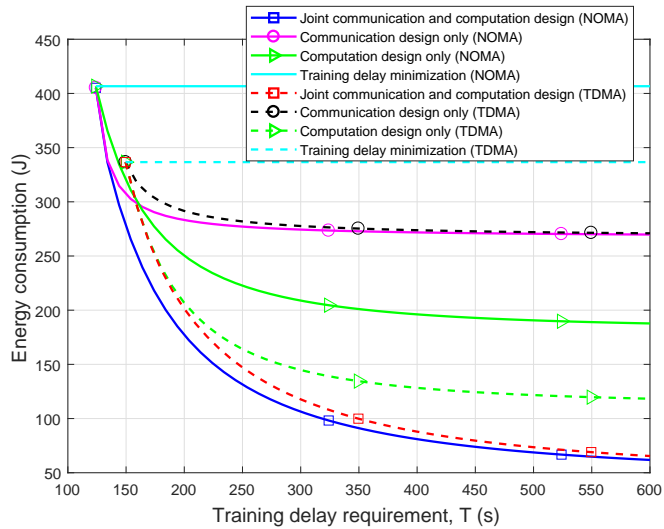}\\
  \small\caption{\small The energy consumption at edge devices versus the training delay requirement $T$.}\label{fig_energy_compared}
\end{figure}

Fig. \ref{fig_energy_compared} shows the energy consumption at edge devices versus the training delay requirement $T$. It is observed that our proposed joint computation and communication designs achieve significant performance gains over the other benchmarks under both NOMA and TDMA transmission protocols. It is also observed that NOMA is feasible when $T$ is larger than $124$ s, while TDMA is only feasible when $T$ is larger than $150$ s. This verifies that NOMA leads to shorter minimum training delay than TDMA, as shown in Remark \ref{outperform}. Furthermore, it is observed that when $T$ is sufficiently large, the performance gap between NOMA and TDMA becomes marginal.

\section{Conclusion}
In this paper, we investigated the energy efficient operation of a federated edge learning system, in which two transmission protocols for ML-parameters uploading, namely NOMA and TDMA, were considered. We minimized the total communication and computation energy consumption at edge devices subject to a given training delay requirement, by jointly designing the communication and computation resource allocations. Although the formulated problems were non-convex, we transformed them into convex forms, and then presented efficient algorithms to solve them optimally. Numerical results revealed interesting tradeoffs among energy consumption, training speed and training accuracy. It was shown that our proposed designs achieve significant performance gains over other benchmark schemes without such joint optimization. It was also shown that properly choosing the numbers of global and local iterations during the distributed training further helps enhance the system energy efficiency by properly balancing the communication-computation energy tradeoff.

\begin{appendices}
\section{Proof of Lemma \ref{f_opt}}
Under any given feasible $\mv \mu$, subproblem (\ref{solve_frequency}) is a convex optimization problem as the objective function is  a convex function with respective to $f_k$ and the constraint $0\leq f_k\leq f_k^{\text{max}}$ is linear for each edge device $k\in\mathcal{K}$. As subproblem (\ref{solve_frequency}) satisfies the Slater's condition, the strong duality holds between subproblem (\ref{solve_frequency}) and its dual problem.

Let $\underline{\tau}_k\geq0$ and $\overline{\tau}_k\geq0$ denote the Lagrange multipliers associated with the inequality $0\leq f_k$ and $f_k\leq f_k^{\text{max}}$, respectively. Then the Lagrangian of $k$-th subproblem in (\ref{solve_frequency}) is given by
\begin{align}
\mathcal{L}_k(f_k,\underline{\tau}_k,\overline{\tau}_k)
=&MN\frac{F_k}{C_k}\varsigma_kf_{k}^2+\frac{F_k}{C_k}\frac{1}{f_{k}}\mu_k
-\underline{\tau}_k f_k+\overline{\tau}_k(f_k-f_k^{\text{max}}).
\end{align}
Let $(\underline{\tau}_k^*,\overline{\tau}_k^*)$ denote the optimal dual solution and $f_k^*$ denote the optimal primal solution to subproblem (\ref{solve_frequency}). They should satisfy the KKT conditions \cite{convex_optimization}, which are given as
\begin{align}
&0\leq f_k^*,f_k^*\leq f_k^{\text{max}}, \underline{\tau}_k^*\geq0,\overline{\tau}_k^*\geq0\label{d}\\
&\underline{\tau}_k^* f_k^*=0, \overline{\tau}_k^*(f_k^*-f_k^{\text{max}})=0\label{e}\\
&2MN\frac{F_k}{C_k}\varsigma_kf_{k}^*-\frac{F_k}{C_k}\frac{1}{{f_{k}^*}^2}\mu_k
-\underline{\tau}_k^*+\overline{\tau}_k^*=0.\label{f}
\end{align}
Based on the above KKT conditions and  via some simple manipulations, this lemma can be easily verified.
\end{appendices}


\end{document}